\documentclass[runningheads]{llncs}
\usepackage[T1]{fontenc}
\usepackage{graphicx}
\usepackage{amsmath,amssymb,amsfonts}
\usepackage{mathtools}
\usepackage{hyperref}
\usepackage{color,soul}

\newcommand{\Real}{\mathbb{R}}

\newcommand{\inv}{{\,\text{-}\hspace{-1pt}1}}

\newcommand{\g}{\mathfrak g}

\renewcommand{\a}{\mathfrak a}
\renewcommand{\k}{\mathfrak k}
\newcommand{\p}{\mathfrak p}

\usepackage{graphicx}
\begin{document}
\title{K--P Quantum Neural Networks}
%
%
\author{Elija Perrier \inst{1}\orcidID{0000-0002-6052-6798} 
}
\authorrunning{E. Perrier}
%
\institute{Centre for Quantum Software and Information, University of Technology, Sydney \email{elija.perrier@gmail.com}}
\maketitle              
\begin{abstract}
We present an extension of K--P time-optimal quantum control
solutions using global Cartan $KAK$ decompositions for geodesic-based solutions. Extending recent time-optimal \emph{constant--$\theta$} control results, we integrate Cartan methods into equivariant quantum neural network (EQNN) for quantum control tasks. We show that a finite-depth limited EQNN ansatz equipped with Cartan layers can
replicate the constant--$\theta$ sub-Riemannian geodesics for K--P problems. We demonstrate how, for certain classes of control problem on Riemannian symmetric spaces, gradient-based training using an appropriate cost
function converges to global time-optimal solutions when satisfying simple regularity conditions. This generalises prior geometric control theory methods
and clarifies how optimal geodesic estimation can be performed in quantum
machine learning contexts.

\end{abstract}
\section{Introduction}
\label{sec:intro}
Time-optimal control of quantum systems is central to many areas of quantum technology,
ranging from fast gate synthesis in quantum computing to high-fidelity pulse shaping
in nuclear magnetic resonance \cite{khaneja_cartan_2001,nielsen_geometric_2006,jurdjevic_geometric_1997}.
The \emph{K--P problem} \cite{jurdjevic_hamiltonian_2001} is a canonical formulation of such time-optimal tasks. K--P problems involve a semisimple Lie algebra $\g$ into $\k\oplus \p$ under a Cartan (or
\emph{involution-based}) decomposition \cite{knapp_lie_1996}. The horizontal controls
come from $\p$, while the compact part $\k=[\p,\p]$ must be generated indirectly via
commutators. Previous work has shown that sub-Riemannian geometry on $\g$ yields geodesics
for locally time-optimal motion \cite{jurdjevic_optimal_2011,boscain_k_2002}. Recently, a new method of optimal control was demonstrated \cite{perrier2024solving} employing global
Cartan $KAK$ decompositions and enforcing a \emph{constant--$\theta$} condition,
one obtains an analytically solvable geodesic for quantum control problems on certain classes of Riemannian symmetric space. Other recent work has examined relaxing the requirement of full equivariance for qubit systems using variational quantum algorithms generated by horizontal elements $\p$ of $KAK$ structures \cite{wiersema2025geometric}. This partial respecting of symmetry translations is equivalent to classes of sub-Riemannian control problems invariant under translations $\p$ but not generators of rotations $\k$. We extend these results in two directions:

\begin{enumerate}
\item \textit{K--P QNN (EQNN) Integration.} \quad
We show how to parameterise the same Cartan-based geodesics using a neural-network ansatz
that is equivariant to the underlying symmetry group. Extending existing work, we show how as with full EQNNs, composing layer-wise networks using our global method which respects the $\k/\p$ partition enables the networks (and importantly, their outputs) to respect sub-Riemannian symmetries. We subsequently explore how K--P respecting networks converge to approximate sub-Riemannian geodesics, without requiring the
user to solve the geodesic equations symbolically.

\item \textit{Time-optimal QNNs} \quad
We demonstrate how QNNs integrating the K--P structure give rise to time-optimal solutions. Specifically we show (i) the existence of a solution such  that finite-depth EQNNs with appropriate
  Cartan layers can exactly represent the constant--$\theta$ geodesics found in \cite{perrier2024solving}; and (ii) the uniqueness of a solution such that that any local optimum of a suitably chosen
  cost function (fidelity plus sub-Riemannian penalty) may, under certain circumstances, converge to the global
  optimum (where target unitaries $U_{target}$ are not in the centralizer of $G$).
\end{enumerate}
Section~\ref{sec:kp} reviews
the K--P setup and the essential Cartan machinery for the constant--$\theta$ approach.
Section~\ref{sec:eqnn} introduces \emph{K--P quantum neural networks} sub-Riemannian layers. Section~\ref{sec:theorems} states and
proves the main theorems on existence and uniqueness for global optimality specific choices of target not in the centralizer of $\g$. Section~\ref{sec:examples}
provides examples and numerical illustrations. Finally, Section~\ref{sec:conclusion} 
offers concluding remarks on open problems.

\section{Background}
\label{sec:kp}

\subsection{The $K-P$ Problem}
The \emph{K--P} problem in control theory is a specific class of problems that seeks, for spaces $G$ decomposable into symmetric $K$ and anti-symmetric $P$ subspaces, the optimal horizontal control $H(t)\!\in\!\p$ with $\|H(t)\|\!\le\!\Omega$ solution that evolves $U(0)=\mathbb I$ to a given $U_{\mathrm{target}}\in G$ in minimum time $T$.  This problem can equivalently be framed as seeking the minimal sub-Riemannian length $\int_{0}^{T}\|H_{\p}(t)\|\,\mathrm{d}t$ over all horizontal curves on the symmetric space $G/K$ satisfying identical endpoint constraints.  Geodesic solutions on the compact symmetric space $G/K$ of the corresponding Schr\"odinger equation can be shown to remain within $K\exp(i\Theta)K$ orbits (which are effectively fixed once the Carten element $\Theta$ is chosen).

\subsection{Cartan decompositions}
Let $G$ be a connected semisimple Lie group (compact for simplicity of exposition) with
Lie algebra $\g$. A \emph{Cartan involution} $\chi$ partitions $\g=\k\oplus \p$, where
$\k$ is the $+1$ eigenspace ($\chi(X)=+X$ for $X\in \k$) and $\p$ the $-1$ eigenspace
($\chi(X)=-X$ for $X\in \p$).  In many quantum applications, $G\subseteq U(n)$ and
$\k$ is the \emph{maximal compact} part while $\p$ is noncompact.  This gives rise to Cartan commutation relations \cite{knapp_lie_1996}:
\begin{align}
    [\k,\k]\subseteq \k, \quad [\p,\p]\subseteq \k, \quad [\k,\p]\subseteq \p. \label{eqn:cartancommutation}
\end{align}
$K=\exp(\k)$ is a subgroup of $G$. If $G$ is compact and semisimple then $K$ is typically a maximal torus. A typical quantum control scenario involves Hamiltonians comprising generators in $\p$, while the evolution generated by the $\k$ part arises via the commutators in (\ref{eqn:cartancommutation}). We want to
implement a target $U_{\mathrm{target}}\in G$ in minimal time subject to an energy cutoff $\|H(t)\|\le \Omega$,
with $H(t)\in \p$.  In sub-Riemannian geometric terms, $\p$ defines the \emph{horizontal}
distribution. $\k$ is the \emph{vertical} direction that can be reached by curvature forms $[\p,\p]\subset \k$ \cite{montgomery_tour_2002}.  
The corresponding group $KAK$ decomposition is given by $G=K\exp(\a)K$,
where $\a\subset \p$ is an abelian subalgebra (\emph{maximally non-compact}
Cartan).  Elements in $\a$ typically look like $i\Theta$, with $\Theta$ real diagonal
in a suitable representation.  Then any $U\in G$ can be written:
\[
   U \;=\; k \,\exp(i\,\Theta)\,c
   \quad\text{for some}\; k,c\in K,\; \Theta\in \a.
\]
\cite{perrier2024solving} showed that the time-optimal solutions 
$\gamma(t)$ can be simplified if we impose $d\Theta(t)=0$ along the path,
the so-called \emph{constant--$\theta$ condition}.
Under such conditions, this yields a closed-form geodesic and the minimal time
$T$ is related to $|\sin(\mathrm{ad}_\Theta)(\Phi)|$, with $\Phi\in \k$ in the commutant
of $\Theta$. The result in \cite{perrier2024solving} can be expressed as follows. Fixing $\dot \Theta = 0$ selects the representative in each $K$-orbit that minimises the Cartan energy functional. This can be expressed in terms of the following theorem. 
\begin{theorem}[Constant--$\theta$ K--P Geodesics]
\label{thm:kpgeod}
Let $\g=\k\oplus\p$ be a Cartan decomposition of a compact semisimple Lie algebra
and let $\Theta\in \a\cap \p$ be in the non-compact Cartan subalgebra.  Suppose
$\Phi\in \k$ commutes with $\Theta$.  Then if $H(t)\in \p$ saturates $\|H\|=\Omega$
and satisfies the \emph{minimal connection} plus $\dot{\Theta}(t)=0$, the minimum time $T$ optimal
path from $U(0)=\exp(\mathrm{i}\Theta)$ to $U(T)=\exp(-\mathrm{i}X)$ 
($X\in \k$ a target) has length
\begin{align}
    \Omega\,T
   \;=\;
   \bigl|\sin(\mathrm{ad}_\Theta)(\Phi)\bigr| \label{eqn:optimaltime}
\end{align}
subject to $X=(1-\cos(\mathrm{ad}_\Theta))(\Phi)$.  
\end{theorem}
This solution leads to:
\begin{align}
   U(t)
   \;=\;
   \exp\bigl(-\mathrm{i}\,\Lambda\,t\bigr)
   \,\sin(\mathrm{ad}_\Theta)(\Phi)
   \,\exp\bigl(+\mathrm{i}\,\Lambda\,t\bigr)
   \quad\text{with}\quad
   \Lambda=\frac{\cos(\mathrm{ad}_\Theta)(\Phi)}{T}. 
\end{align}
See \cite{perrier2024solving,perrier2024quantum} for proofs and exposition. Theorem \ref{thm:kpgeod} shows that for certain classes of quantum control problem corresponding to symmetric spaces, there exist globally optimal controls can be solved analytically. Below, we adapt and integrate Theorem \ref{thm:kpgeod} into a quantum neural network setting.

\section{K--P QNNs}
\label{sec:eqnn}
 EQNNs 
architect neural networks to respect group symmetries in ways that facilitate task optimisation \cite{nguyen_theory_2022,ragone_representation_2023} and have shown success in quantum optimisation tasks. EQNNs have layers designed so that transformations by a group $G$ act consistently on 
inputs and outputs \cite{larocca_group-invariant_2022}. The K--P problem has an explicit decomposition $\g=\k\oplus \p$ associated with
the involution $\chi$.  A route to building an QNN respecting K--P structure is to make 
the layer transformations equivariant with respect to symmetry subgroup $K$ (the subgroup
generated by $\k$).  In this formulation,  unitary conjugation
$k\in K$ transforms the QNN parameters consistent with Eqn. (\ref{eqn:cartancommutation}) above. In the simplest sense, an K--P layer can be written:
\begin{equation}
\label{eq:eqnnlayer}
   \mathcal{L}_{\boldsymbol{\theta}}(\rho)
   \;=\;
   \exp\bigl(-\mathrm{i}\,H_{\mathrm{p}}(\boldsymbol{\theta})\bigr)
   \,\rho\,
   \exp\bigl(+\mathrm{i}\,H_{\mathrm{p}}(\boldsymbol{\theta})\bigr),
\end{equation}
where $H_{\mathrm{p}}(\boldsymbol{\theta})\in \p$ is restricted to the horizontal
subalgebra $\p$.  Conjugation by $\k$ then implements a vertical shift, but is consistent
with the underlying symmetry (cycling generators within $\k$ and $\p$ respectively).  The key relation is set by $[\p,\p]\subset \k$ which allows generalised rotations to be synthesised in a controlled manner using generators in $\p$, we can keep the distribution structure.In K--P tasks we typically want controls only in $\p$.  To replicate the constant--$\theta$ geodesics of \cite{perrier2024solving}, we construct layer Eq.~\eqref{eq:eqnnlayer} in a way that fixes $\Theta$ and so that the net effect on $\k$ arises from commutators $[\p,\p]$.  In \cite{perrier2024solving}, the minimal connection is given by: 
\[
   k^\inv dk
   \;=\;
   -\cos(\mathrm{ad}_\Theta)(dc\,c^\inv),
\]
and hence $U(t)$ could be generated purely by $\p$.  We encode such constraints
into the QNN layer as follows: 
\begin{enumerate}
\item \textit{Initialise $\Theta$.}  First, we choose an element $\Theta\in \a\cap\p$
(so it is in the noncompact Cartan subalgebra).  In an $n$-qubit representation,
this might be a block diagonal or simple diagonal with real entries.

\item \textit{Parameterise commutant $\Phi\in \mathrm{Comm}(\Theta)\cap \k$.}
   Because $\Phi$ commutes with $\Theta$, $\mathrm{Ad}_{\exp(\mathrm{i}\Theta)}(\Phi)=\Phi$.

\item \textit{Generate horizontal pulses.}  The net effect of turning on $\Phi$ plus
$\sin(\mathrm{ad}_\Theta)(\Phi)\in \p$ replicates the constant--$\theta$ geodesic. The layer exponentiates $\sin(\mathrm{ad}_\Theta)(\Phi)$
for a certain amplitude $\alpha$, while also exponentiating $\Phi$ for a
turning rate $\lambda$.

\item \textit{Repeat in a multi-layer QNN.}  Several such layers can be stacked or
interleaved with standard universal gates.  If the objective is to achieve
the final $U_{\mathrm{target}}$ with minimal $\|H\|\,T$, the network can be trained
via a cost function $C=1-\mathrm{fidelity}+\kappa\cdot(\text{time penalty})$,
just as in typical VQA approaches \cite{cerezo_variational_2021}.
\end{enumerate}

Because the layer is built from $\{\Theta,\Phi\}$ with $\Theta$ fixed,
the QNN is automatically equivariant under transformations in $K$ that fix
$\Theta$ (or map it to an isomorphic subalgebra).  In practice, the numerical training need not solve the entire geodesic system explicitly, but the final result (provided local minima are avoided) ought to match the sub-Riemannian solution (that is, a global length-minimising horizontal curve).

\section{Existence and uniqueness of K--P circuits}
\label{sec:theorems}
We now set out results showing (1) the existence of a finite-depth EQNN with
Cartan layers can represent the constant--$\theta$ K--P solution, and
(2) the uniqueness of cost function minima via convergence with global minima once sub-Riemannian constraints
are imposed.  First, we show the existence of a finite-depth EQNN circuit for the constant-$\theta$ solution.

\begin{theorem}[K--P QNN Circuit (Existence)]
\label{thm:expr}
Consider a quantum system whose algebra $\g$ has a Cartan decomposition $\k\oplus \p$.
Let $\Theta\in \a\cap \p$, and let $\Phi\in \k$ commute with $\Theta$.  Then there
is a finite-depth EQNN ansatz,
\begin{equation}
\label{eq:eqnnansatz}
  U(\boldsymbol{\alpha})
  \;=\;
  \prod_{j=1}^m
    \exp\bigl(\mathrm{i}\,\alpha_{j}^{(\p)}\,Y_{j}^{(\p)}\bigr)
    \;\exp\bigl(\mathrm{i}\,\alpha_{j}^{(\k)}\,Y_{j}^{(\k)}\bigr),
\end{equation}
with $Y_{j}^{(\p)}\in\p$, $Y_{j}^{(\k)}\in\k$ and controls $\alpha \in \Real^n$ (for $n=\dim \g$), that can exactly 
realize the constant--$\theta$ solution in Theorem~\ref{thm:kpgeod} for any choice
of $\Theta,\Phi$. In particular, $U(\boldsymbol{\alpha}^*)=U(T)$ from
Theorem~\ref{thm:kpgeod} for some $\boldsymbol{\alpha}^*$.
\end{theorem}

\begin{proof}
As $\Phi\in \mathrm{Comm}(\Theta)$, then $\mathrm{Ad}_{\exp(\mathrm{i}\Theta}(\Phi)=\Phi$,
and $\mathrm{ad}_\Theta(\Phi)=[\Theta,\Phi]$ remains in $\p$.  The Trotter expansions
of $\exp(\sin(\mathrm{ad}_\Theta)(\Phi))$ can be compiled into a finite product of
exponentials in $\p$ and $\k$.  Letting $Y^\p_1=[\Theta,\Phi]\in\p$, we see
that $\sin(\mathrm{ad}_\Theta)(\Phi)$ is a linear combination of repeated commutators
of $Y^\p_1$ with $\Theta$.  Repeated commutators remain in $\p$ or $\k$ by the standard Cartan relations. Hence, we can approximate $\exp\bigl(\mathrm{i}\,\sin(\mathrm{ad}_\Theta)(\Phi)\bigr)$
as a product of exponentials of $\p$ or $\k$.  Appending a single factor
$\exp(\mathrm{i}\,\Lambda\,t)$ in $\k$ recovers the $\cos(\mathrm{ad}_\Theta)(\Phi)$
term.  Thus, \eqref{eq:eqnnansatz} suffices to represent $U(T)$. 
The condition that $\boldsymbol{\alpha}\mapsto (Y_j^{(\p)},Y_j^{(\k)})$
behaves in an equivariant manner (i.e.\ transforms consistently under $k\in K$)
follows from choosing the $Y_j$ in each layer to be sums of $\p$ or $\k$ sub-blocks, respecting $\chi$.
\qed
\end{proof}
Because the generators generate a dense subgroup of $G$ (by the Cartan decomposition), the Lie-Trotter approximation in principle allows approximation to arbitrary accuracy up to Trotter error $\varepsilon$ using a finite product of $\p$ and $\k$ exponentials. The control parameters $\alpha$ may be learnt and optimised according to a typical VQA optimisation algorithm. Next, we show that any local optimum of a standard cost function in the EQNN approach must coincide with the global sub-Riemannian geodesic solution. 
%
%
\begin{theorem}[K--P Stationary Points: Uniqueness]\label{thm:KPuniqueness}
Let \,$U(\boldsymbol{\alpha})$ be a finite-depth circuit of the form \eqref{eq:eqnnansatz}, 
where each layer is generated by operators in $\p$ (the ``horizontal'' subalgebra) 
with or without additional commutator-generated rotations from $\k$. 
Define the cost function
\begin{equation}\label{eq:CostFunction}
   C(\boldsymbol{\alpha})
   \;=\;
   \bigl[\,
      1 \;-\;\mathrm{Fid}\bigl(U(\boldsymbol{\alpha}),\,U_{\mathrm{target}}\bigr)
   \bigr]
   \;+\;
   \gamma\,
   \Bigl(\,\mathcal{L}(\boldsymbol{\alpha}) \;-\;\Omega\,T\,\Bigr)^{2} 
\end{equation}
where \,$\mathrm{Fid}(\cdot,\cdot)$ is a fidelity-like measure between unitaries
(e.g.\ $1-\mathrm{Re}\{\mathrm{Tr}[U^\dagger(\boldsymbol{\alpha})\,U_{\mathrm{target}}]\}$),
$\gamma>0$ is a weighting constant, and 
\,$\mathcal{L}(\boldsymbol{\alpha})=\int_{0}^{T}\!\|\!H_{\p}(t)\|\!\;dt$ 
represents the sub-Riemannian path length of the trajectory 
\,$U(\boldsymbol{\alpha})$ if \,$\|H_{\p}(t)\|\le \Omega$. 
Suppose that 
\begin{equation}\label{eqn:StationaryCondition}
   \nabla_{\!\boldsymbol{\alpha}}\,C(\boldsymbol{\alpha}^*) \;=\; 0
   \quad\text{and}\quad
   \|H_{\p}(t)\|\le\Omega \;\,\forall\,t\,.
\end{equation}
Then $U(\boldsymbol{\alpha}^*)$ coincides with the constant--$\theta$ geodesic from 
Theorem~\ref{thm:kpgeod} and is (globally) time-optimal under the K--P constraints.
\end{theorem}

\begin{proof}
By Pontryagin’s Maximum Principle and standard sub-Riemannian geometry arguments
(see \cite{jurdjevic_geometric_1997,boscain_k_2002} for details), 
the unique path of minimal $\p$-length $\int_0^T \|H_{\p}(t)\|\>dt$ 
subject to $\|H_{\p}(t)\|\le\Omega$ and $U(0)=\mathbb{I},\,U(T)=U_{\mathrm{target}}$ 
must satisfy the minimal connection equations and yield the constant--$\theta$ 
solution \eqref{eqn:optimaltime} in Theorem~\ref{thm:kpgeod}.
Since \eqref{eq:CostFunction} strictly penalizes both infidelity and any suboptimal 
$\p$-length $\mathcal{L}(\boldsymbol{\alpha})-\Omega T$, 
a stationary point $\boldsymbol{\alpha}^*$ with zero gradient 
\eqref{eqn:StationaryCondition} cannot be a spurious local minimum 
unrelated to the sub-Riemannian geodesic. 
Hence the only possible local minimiser is the globally optimal 
constant--$\theta$ solution. 
Because \eqref{eq:eqnnansatz} shows such a solution is exactly representable
by the finite-depth circuit, it follows that $U(\boldsymbol{\alpha}^*)$ 
must coincide with the time-optimal geodesic. 
\qed
\end{proof}
Because the optimal Hamiltonian here is normal (otherwise failing to satisfy the maximum principle), abnormal trajectories are excluded. Stationary points with $\mathcal{L} \ge \Omega T$ are precluded by assumptions that $\nabla_{\boldsymbol{\alpha}} C = 0$ and $\lambda(\mathcal{L} - \Omega T) = 0$ and $\lambda \geq 0$. Theorems \ref{thm:expr} and \ref{thm:KPuniqueness} show that (for certain choices of subgroup $K$ (and thus symmetric space) that K--P QNNs can be made both complete and globally convergent: (i) the finite-depth circuit is guaranteed to express a minimal-time path; and (ii) gradient-based optimisation (with an appropriate cost) can avoid spurious local minima. We briefly illustrate how K--P QNNs with a can discover the same time-optimal solution as an analytic or geometric approach.

\section{Evaluation: $\Lambda$-systems and $SU(3)$}
\label{sec:examples}
Consider the three-level $\Lambda$-system studied in \cite{perrier2024solving} where results from \cite{albertini_sub-riemannian_2020} were reproduced using the constant-$\theta$ method. \(\mathrm{SU}(3)\), has a representation in terms of Gell-Mann generators which can be decomposed into a horizontal subalgebra
\(\p\subset\g\) (those used as direct controls) and a vertical subalgebra \(\k\subset\g\) (generated indirectly via commutators). The $\Lambda$ system is
a three-level model with horizontal transitions $\p$ coupling two ground states 
to an excited state. To evaluate our K--P QNN, we construct python code to test its efficacy at reproducing the $\Lambda$-system optimisation results in \cite{perrier2024solving}. We select:
\[
   \k \;=\;\mathrm{span}\bigl\{\,{-i\lambda_3},\,{-i\lambda_6},\,{-i\lambda_7},\,{-i\lambda_8}\bigr\},
   \quad
   \p \;=\;\mathrm{span}\bigl\{\,{-i\lambda_1},\,{-i\lambda_2},\,{-i\lambda_4},\,{-i\lambda_5}\bigr\}.
\]
We construct a finite-depth product of exponentials:
\[
   U(\boldsymbol{\alpha}) 
   \;=\;
   \prod_{j=1}^L 
      \exp\!\Bigl(\mathrm{i}\!\sum_{a\in\k}\alpha_{j,a}^{(\k)}\,Y_a\Bigr)\;
      \exp\!\Bigl(\mathrm{i}\!\sum_{b\in\p}\alpha_{j,b}^{(\p)}\,X_b\Bigr).
\]
Here, \(\{Y_a\}\subset\k\) and \(\{X_b\}\subset\p\) are the basis elements.  We implement this as the function 
\verb|circuit_forward|, which sequentially multiplies each layer's matrix exponential 
(see Theorem~\ref{thm:expr}). From Theorem~\ref{thm:KPuniqueness} (cf.~Eq.~\eqref{eq:CostFunction}), 
we define a cost that has two main terms: 
\[
   C(\boldsymbol{\alpha})
   \;=\;
   \bigl[\,1-\mathrm{Fid}\bigl(U(\boldsymbol{\alpha}),\,U_{\mathrm{target}}\bigr)\bigr]
   \;+\;
 \gamma\,\bigl(\text{path\_length}(\boldsymbol{\alpha})\bigr).
\]
In the code, we approximate the sub-Riemannian path length by the sum of the \(\ell_{2}\)-norms of 
the horizontal parameters \(\alpha^{(\p)}\) encoded in \verb|path\_length|.  
This is a simplified version of \(\int_{0}^{T}\|H_{\p}(t)\|\mathrm{d}t\), sufficient to demonstrate
the principle of penalizing the magnitude of \(\p\)-controls.
We apply a straightforward gradient descent (using JAX auto-differentiation).  By Theorem~\ref{thm:KPuniqueness}, 
no spurious local minima exist if \(\boldsymbol{\alpha}\) saturates the bracket generation assumptions and 
\(\|H(t)\|\le\Omega\).  Convergence to fidelity \(\approx1\) and small path length indicates that the learned 
solution reproduces the sub-Riemannian geodesic.  
\begin{figure}
    \centering
\includegraphics[width=0.75\linewidth]{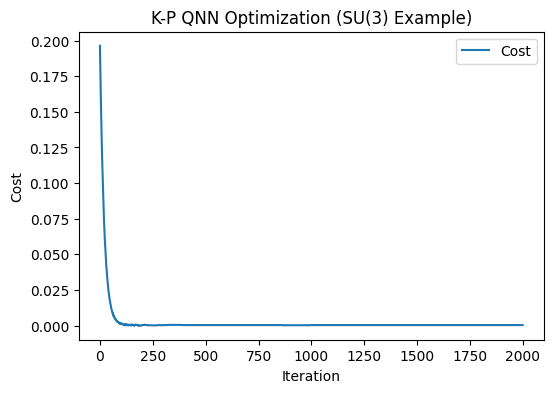}
    \caption{K--P QNN Loss function (Eqn. (\ref{eq:CostFunction}) showing convergence between $U(\alpha)$ and $U_{target}$. Convergence of both terms shows that the K--P QNN learns the optimal unitary parametrised by controls $\alpha$ (reproducing targets in \cite{albertini_sub-riemannian_2020} and \cite{perrier2024solving}).}
    \label{fig:kpqnn-loss}
\end{figure}
As shown in Fig. (\ref{fig:kpqnn-loss}), after sufficient epochs, the cost function converges close to zero, and the final unitary 
\(U(\boldsymbol{\alpha}^{*})\) attains \(\mathrm{Fid}\approx1\) (specifically 0.9998576).  
Hence, the constant--\(\theta\) time-optimal solution \(\exp\bigl(-\tfrac{\mathrm{i}\pi}{4}\,\lambda_6\bigr)\) 
is accurately reconstructed, in agreement with the analytical results (cf.\ Section~4 and Eq.~\eqref{eqn:optimaltime}).  
Thus the K--P QNN approach recovers the same geodesic solution from a purely data-driven perspective. The repository is available at \verb|https://github.com/eperrier/k-p_qnn|.

\section{Conclusions and Outlook}
\label{sec:conclusion}
We have shown how K--P QNNs, a form of EQNNs can
be constructed using the \emph{constant--$\theta$} Cartan decomposition
approach to the K--P problem.  The synergy arises from the geometric consistency:
EQNN layers that keep $\Theta$ fixed and use pulses in $\p$ effectively trace out
the same sub-Riemannian geodesic described by \cite{perrier2024solving}.  Our results show that local optimality in a typical QNN variational cost function 
indeed implies global time-optimality for certain classes of quantum control problems on symmetric spaces $G/K$. Limitations of our method include those set out in \cite{perrier2024solving}, particularly that time-optimal sequences are found for only certain targets in $\g$. We have also assumed no complications arising from cut or critical loci. Future research directions building on this work may consider:
\begin{itemize}
\item incorporating noise and decoherence by letting the cost function measure fidelity under realistic noise channels;
\item extending to non-compact or indefinite metrics (e.g. involving indefinite Killing forms or non-compact groups); 
\item examining how critical and cut loci (local component Maxwell sets, conjugate points) affect numerically-obtained local minima;
\item experimental demonstration via implementation of EQNNs using NISQ superconducting devices. 
\end{itemize}
Our work contributes to the growing literature connecting Cartan-based geodesic solutions to quantum machine learning protocols. 

\subsubsection{Acknowledgements} The author thanks Dr Chris Jackson for his considerable mentorship and discussions on the topic. This work was independently funded by the author. The author declares that they have no conflicts of interest.

\bibliographystyle{splncs04}
\bibliography{refs,referencesfinal3}

\end{document}